\documentclass[runningheads]{llncs}

\usepackage[utf8]{inputenc} % allow utf-8 input
\usepackage[T1]{fontenc}    % use 8-bit T1 fonts
\usepackage{hyperref}       % hyperlinks
\usepackage{url}            % simple URL typesetting
\usepackage{booktabs}       % professional-quality tables
\usepackage{amsfonts}       % blackboard math symbols
\usepackage{nicefrac}       % compact symbols for 1/2, etc.
\usepackage{microtype}      % microtypography

\usepackage{amssymb}
\usepackage{amsmath}
\usepackage[inline]{enumitem}
\usepackage{algorithm}
\usepackage[noend]{algpseudocode}
\usepackage{cleveref}
\usepackage{graphicx}
\usepackage{caption}
\usepackage{subcaption}

\graphicspath{ {images/} }

\DeclareMathOperator*{\argmin}{\arg\!\min}

\usepackage[left=2cm, right=2cm, top=2cm, bottom=2cm]{geometry}

\title{Efficient Processing of Subsequent Densest Subgraph Query}

\author{%
  Chia-Yang Hung and Chih-Ya Shen\\
  \texttt{s110062522@m110.nthu.edu.tw, chihya@cs.nthu.edu.tw} \\
}

\institute{Department of Computer Science\\
  National Tsing Hua University\\}
\begin{document}

\maketitle

\begin{abstract}
  Dense subgraph extraction is a fundamental problem in graph analysis and data mining, aimed at identifying cohesive and densely connected substructures within a given graph. It plays a crucial role in various domains, including social network analysis, biological network analysis, recommendation systems, and community detection. However,  extracting a subgraph with the highest node similarity is a lack of exploration. To address this problem, we studied the Member Selection Problem and extended it with a dynamic constraint variant. By incorporating dynamic constraints, our algorithm can adapt to changing conditions or requirements, allowing for more flexible and personalized subgraph extraction. This approach enables the algorithm to provide tailored solutions that meet specific needs, even in scenarios where constraints may vary over time. We also provide the theoretical analysis to show that our algorithm is $\frac{1}{3}$-approximation. Eventually, the experiments show that our algorithm is effective and efficient in tackling the member selection problem with dynamic constraints.
\end{abstract}

\section{Introduction}
Graph-based data mining has become a prominent research focus in the field of data science. With the need to represent data existence using graph structures in domains such as biology, chemistry, social networks, and financial systems, graph mining holds immense value. The task of finding or labeling specific types of subgraph structures within a graph is a fundamental problem in discrete mathematics and finds extensive applications in the field of graph mining.\\
The utilization of graph mining is widespread across various domains. In biology, it can be used to analyze protein-protein interaction networks or gene regulatory networks, enabling the identification of functional modules or disease-related pathways. In chemistry, graph mining techniques are employed to discover molecular structures or chemical compounds with specific properties. In social networks, graph mining aids in understanding the dynamics of communities, identifying influential users, and detecting patterns of interaction. In financial systems, it can help analyze complex networks of transactions and identify fraudulent activities or systemic risks.\\
Among the structures analyzed by researchers, there is a problem known as the Member Selection problem (MSP). The objective of this problem is to develop an algorithm that can find, within polynomial time, an induced subgraph in a given graph with the highest average similarity among its nodes. However, it has been proven that this problem is NP-hard and inapproximable, meaning it is computationally challenging to find an optimal solution in a reasonable amount of time.\\
The NP-hardness of the subgraph selection problem implies that it belongs to a class of computationally difficult problems for which no known efficient algorithm exists. As a result, researchers often resort to heuristic or approximation algorithms to tackle this problem. These algorithms aim to provide reasonably good solutions within a reasonable amount of time, even though they may not guarantee an optimal solution.\\
Indeed, approximation algorithms have been proposed that demonstrate the ability to find approximate solutions for the MSSG problem. The algorithm offers solutions that are reasonably close to the optimal solution. However, in the long run, repeatedly running the same algorithm for different requirements on the same graph may lead to unnecessary time wastage. Therefore, it is desirable to minimize the time overhead associated with running the algorithm for different needs on the same graph. For instance, consider a social networking platform that aims to form interest-based communities or discussion groups for its users. Each host has specific preferences and constraints regarding the size and composition of the group they want to be a part of. However, these preferences can change over time as hosts discover new interests or undergo personal growth. Additionally, the platform may introduce new features or algorithms to enhance the group formation process, leading to evolving constraints and requirements. In such a dynamic environment, the Dynamic Constraint Member Selection Problem (DCMSP) becomes crucial. DCMSP focuses on adapting to changing size and similarity constraints when forming these interest-based groups. By incorporating dynamic constraint management algorithms, the platform can dynamically adjust the group composition to match the evolving preferences and constraints of its users.

The contributions of this paper can be summarized as follows:
%\vspace{-10pt}
\begin{itemize}
    \item The Problem Formulation: This paper introduces the Dynamic Constraint Member Selection Problem (DCMSP), which addresses the challenges posed by diverse requirements and complexities. DCMSP aims to tackle the ever-changing demands and constraints in the process of member selection.
    \item Complexity Analysis: To tackle the efficiency of dynamic need on different situation, our algorithm has been proved that it is also a $\frac{1}{3}$-approximation algorithm.
    \item Experimental Evaluation: The paper conducts experiments using real datasets to evaluate the performance of the proposed algorithm and compare the computation time to other baselines. Through rigorous evaluation and comparison, the algorithm demonstrates superior performance and effectiveness in addressing the problem at hand.
\end{itemize}

\section{Related Works}
\subsection{Dense Subgraph Extraction}
 In recent years, the extraction of dense subgraphs from large graphs plays a crucial role in various application domains. There has been significant interest in studying various complex networks such as the World Wide Web, social networks, and biological networks. In the context of the Web graph, dense subgraphs can indicate thematic groups or even spam link farms \cite{10.1145/1244408.1244417}. In the financial domain, one application of extracting dense subgraphs includes the identification of price value motifs. This approach has been employed among other methods for analyzing financial data \cite{10.1145/1557019.1557142}. In order to confront this challenge, Huang et al. address the problem of searching for the closest community using a $k$-truss based community model \cite{huang2015approximate}. Li et al. investigate the extraction of a set of $k$-core subgraphs aiming to maximize the minimum node weight \cite{10.14778/2735479.2735484}. In addition, a set of dense subgraph extraction problems were integrated with various factors, such as spatial~\cite{yang2012socio,shen2015socio,shen2020activity}, temporal~\cite{chen2018efficient}, skills~\cite{yang2021learning,shen2017finding}, and many others~\cite{huang2009provable,hsu2020wmego,shuai2013pattern,yang2022enhancing}.
 As a result, the problem of extracting dense subgraphs has garnered substantial attention.
% \subsection{Member Selection for Online Support Group}

\subsection{Densest Subgraph Problem and Densest k-subgraph Problem}
The Densest Subgraph Problem (DSP) is a branch of the broader field of dense subgraph extraction. In DSP, the focus is specifically on identifying the subset of vertices that maximizes the degree density. Spefically, DSP is a well-known formulation that seeks to identify a subset of vertices $S \subseteq V$ that maximizes the degree density of $S$. The degree density, denoted as $d(S)$, is defined as the ratio of the number of edges within $S$, denoted as $e[S]$, to the cardinality of $S$, denoted as $|S|$. The objective is to find the subset of vertices that achieves the highest degree density $d(S)=\frac{e[S]}{|S|}$. Andrew V. Goldberg's groundbreaking contribution was the development of a polynomial time algorithm that utilizes a max flow technique to identify the maximum density subgraph \cite{goldberg1984finding}. Charikar's work made a significant contribution to the advancement of fast algorithms for the Densest Subgraph (DSP) \cite{10.5555/646688.702972}. His research introduced influential developments that have greatly impacted the efficiency and effectiveness of algorithms designed for solving these problems. \\
On the other hand, size-constrained versions of the DSP have been extensively studied in the literature, as they offer practical utility in applications where the size of the solution must be controlled. Variants of DSP exist that impose specific size requirements on the output containing exactly k nodes. In formulation, densest $k$-subgraph problem intent to find a induced subgraph $S$ that maximizes $d(S)=\frac{e[S]}{|S|}$ subject to $|S|=k$. Although DSP itself can be solved in polynomial time, the introduction of size constraints makes the problem computationally challenging. U. Feige et al. developed a polynoimal time algorithm which approximation ratio is $n^\delta$, where $\delta$ is a constant approximately equal to $\frac{1}{3}$ \cite{feige2001dense}. Bhaskara et al. have proposed another algorithm that achieves an approximation ratio of $\mathcal{O}(n^\frac{1}{4})$. Their algorithm is also inspired by studying an average case version of the problem, specifically focusing on differentiating between random graphs and random graphs with planted dense subgraphs. In specific, their algorithms involve a clever approach of counting appropriately defined trees of constant size within the graph $G$. By utilizing these counts, they are able to identify the vertices that belong to the dense subgraph. These advancements contribute to the efficiency and effectiveness of solving the size-constrained DSP.

\section{Preliminaries}\label{Preliminaries}
In this section, we first provide some background material on graphs and the Member Selection Problem.
Let one heterogeneous graph $G=(V,E,S)$, where $V=\{v_1, v_2, \cdots, v_n\}$ is the set of $n$ nodes, $E\subseteq V\times V$ is the set of edges, and $S$ represent the similarity edges between two nodes. A weight, denoted as $w\in (0,1]$, is assigned to each similarity edge $s\in S$. In general, given an attributed graph $G=(V,E,X)$, where $X\in\mathbb{R}^{n\times d}$ is the attributed matrix. We can model the $w[u,v]=w[v,u]=\sqrt{\sum_{\forall i}(1-|X_{u,i}-X_{v,i}|)^2}$. Squared Euclidean Distance has been proved to be useful in grouping node according to their attributes \cite{abramowitz2003symptom}. Next, we introduce Member Selection Problem (MSP) proposed by shen et al \cite{10.1145/2806416.2806423}. MSP aims to identify an induced subgraph $F \subseteq G$ that adheres to the following principles:
\begin{enumerate}
    \item Maximize the average similarity between all nodes in $F$.
    \item Ensure the group size $|F| > p$, where $p$ is the predefined size constraint.
    \item Every two nodes in $F$ has edge, i.e., $F$ is a complete graph. 
    \item Ensure that every node in $F$ has a least one weight edge is larger than similarity constraint $s$.
\end{enumerate}
For the first one of these, we first define the total weight function $W:G\rightarrow\mathbb{R}$ by $\forall F \subset G$ $$W(F) = \sum_{u,v\in F}w[u,v]$$
and next define the average similarity function $\alpha:F\rightarrow\mathbb{R}$ by $\alpha(F)=\frac{W(F)}{|F|}$. And the third principles can be formulated as $\forall v\in F, \exists u\in F$ such that $w[u,v]=w[u,v]>s$. Member Selection problem has been proves that it is a NP-Hard problem and inapproximable with any factor unless $P=NP$ \cite{10.1145/2806416.2806423}. 

% \begin{algorithm}
% \caption{SGSEL}
% \scriptsize
% \begin{algorithmic}[1]
% \Require Graph $G=(V, E, S)$, similarity constraint $s$, and size constraint $p$
% \Ensure 3-Approximation solution $F$
% \State $V\leftarrow A$ $\triangleleft$ We remove vertex v from V if all of its incident similarity edges have edge weights smaller than s.
% \State $T\leftarrow V$ $\triangleleft$ T represents the set of vertices that have not been selected as reference vertices yet.
% \State $F^{APX} \leftarrow \varnothing$ $\triangleleft$ Initialize the approximation solution $F$
% \While{ $T \neq \varnothing$}
%     \State $v\leftarrow \argmin_{u\in T} I_G(u)$ $\triangleleft$ $I_G(u)$ is defined as the sum of the edge weights between vertex $u$ and all other vertices $t$ in $G$.
%     \State $T\leftarrow T \setminus \{v\}$
%     \State $C_v \leftarrow \{v\} \cup \mathcal{N}(v)$ 
%     \If {$|C_v| < p$}
%         \State continue;
%     \EndIf
%     \State Let $\Gamma_1 \leftarrow |C_v|$
%     \For {$i \leftarrow 1$ to $|C_v|$}
%         \State $\hat{v_i} \leftarrow \argmin_{u\in\Gamma_i}$ $I_{\Gamma_i}(u)$
%         \State $\Gamma_{i+1} \leftarrow \Gamma_i \setminus \{\hat{v_i}\}$
%     \EndFor
%     \State Let $\Gamma_v^*$ be the $\Gamma_i$ with maximum $\alpha(\Gamma_i)$ and $|\Gamma_i| \geq p$, $\forall i$
%     \If{$\alpha(\Gamma_v^*) > \alpha(F^{APX})$}
%         \State $F^{APX} \leftarrow \Gamma_v^*$
%     \EndIf
% \EndWhile
% \State PostProcessing($F^{APX}$)
% \State \Return $F^{APX}$
% \end{algorithmic}
% \end{algorithm}

Although MSP is inapproximable with any factor, Shen et al. still provide a $\frac{1}{3}$-approximation solution called Support Group Member Selection (SGSEL) when we relaxed the constraint of complete graph in $F$ \cite{10.1145/2806416.2806423}. SGSEL extracts a subgraph $F$ starting from a 2-clique $C_v$ with reference node $v$. Next, in the second step, SGSEL iteratively eliminates vertices within each $C_v$ that have smallest incident similarity. To address the complete requirement of the Member Selection problem, SGSEL propose a post-processing procedure that customizes the 2-clique solution $F^{APX}$. This procedure ensures that the selected members in $F^{APX}$ meet the desired complete criteria, further enhancing the objective value $\alpha$ of the MSP solution. As mention above, it has been theorical proved that SGSEL is a 3-approximation solution on member selection problem without complete graph constraint. That is, let $F^{SGSEL}$ be an approximation solution identity by SGSEL and $F^{OPT}$ be the optimal solution in MSP, we have $\alpha(F^{SGSEL}) > \frac{\alpha(F^{OPT})}{3}$. \\

\begin{figure}
    \centering
    \includegraphics[scale=0.5]{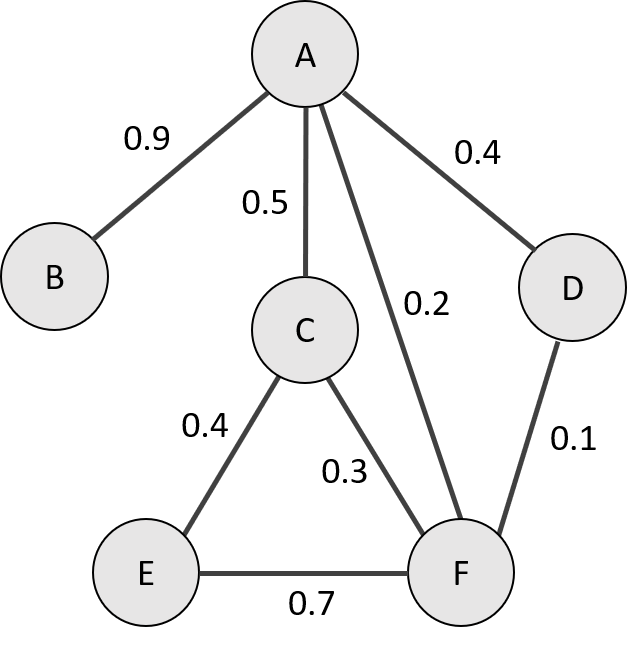}
    \caption{Example for SGSEL}
    \label{fig:example_1}
\end{figure}

\textbf{Example.} Consider the example in figure \ref{fig:example_1}. SGSEL iteratively selects nodes as reference node in graph. After SGSEL choosing CC as reference node and expand its 2-clique $C_v=\{A,C,E,F\}$ as $\Gamma_1$. In first iteration, AA is chosen to remove because $A$ has minimum $I_{\Gamma_1}(u)$, $\forall u \in \Gamma_1$. Thus, $\Gamma_2 = \{C,E,F\}$. In next iteration, SGSEL remove $C$ since $I_{\Gamma_2}(C)=0.7$ is the smallest. This procedure repeats until $\Gamma_4 = \varnothing$. Then SGSEL take $\Gamma_1$ as $\Gamma_C^*$ which has maximum objective value $\alpha(\Gamma_1)=0.525$, $\forall \Gamma_i$. Afterward, SGSEL choosing next reference node DD and continue above process. Eventually, SGSEL return the solution $F^{APX} = \{A,C,E,F\}$.

\section{Problem Formulation}
Due to the design of the MSP algorithm, it fails to consider the process under changing conditions. Moreover, as the number of requirements increases, it becomes evident that we need more flexible algorithms to effectively address the challenges presented by different analysis scenarios. Thus, we describe the Dynamic Constraint Member Selection Problem (DCMSP). DCMSP is focus on the change of size or similarity constraint. Moreover, we present the Extended Dynamic Constraint Member Selection Problem (EDCMSP). EDCMSP incorporates the idea of divide and conquer, breaking down DCMSP into two subgraphs, to expedite the expansion of solutions for DCMSP. In this paper, we will discuss the part of size constraint $p$. 
\subsection{DCMSP}
In order to better describe the dynamic constraints of parameter $p + \Delta p$, it is necessary to introduce a more comprehensive and flexible framework. Consider $G=(V,E,S)$, the objective of DCMSP is to find an induced subgraph $F\subset G$ that follow to the following three principles:
\begin{enumerate}
    \item Maximize the average similarity $\alpha(F)$ between all nodes in $F$.
    \item Ensure the group size $|F| > p+\Delta p$, where $p$ is the size constraint and $\Delta p$ is the \textbf{dynamic} variable.
    \item Ensure that every node in $F$ has a least one weight edge is larger than similarity constraint $s$.
\end{enumerate}
As the parameter $\Delta p$ decreases, it corresponds to a relaxation of the overall constraints. Therefore, there is no need for further discussion regarding this scenario.
\subsection{EDCMSP}
In EDCMSP, We can approach this problem from a differnet perspective. If we have already solved the member selection problem in graph $G$ with size constraint $p$ and identified an approximate solution $F$. We consider the remaining portion of the original graph $G\setminus F$ as a newly added graph. And next, treat the new constraints $\Delta p$ as constraints specific to this newly added graph. In general, we can re-formulate this perspective as: Given two disjoint heterogeneous graph $G_1, G_2$ with same similarity constraint $s$ and different size constraint $p_1, p_2$, respectively. We intent to find a induced subgraph $F' \subseteq G_1 \cup G_2 $ such that 
\begin{enumerate*}[label=(\roman*)]
    \item Maximize the average similarity $\alpha(F)$ between all nodes in $F'$.
    \item Ensure the group size $|F'| > p_1+p_2$.
    \item Ensure that every node in $F$ has a least one weight edge is larger than similarity constraint $s$.
\end{enumerate*}
Note that $F$ and $G\setminus F$ might have some edges, but $G_1$ and $G_2$ not. To tackle this, we might manually add some edges between $G_1$ and $G_2$.

\section{Algorithm Design}
Although SGSEL has been proven to be able to identify an induced subgraph that serves as a $\frac{1}{3}$-approximation solution for the member selection problem \cite{10.1145/2806416.2806423}, there are certain limitations and considerations that need to be taken into account. Here we propose the algorithm called Dynamic Constrain member SELecition (DCSEL) aim to solve the DCMSP efficiently. DCSEL first extract an induced subgraph $F_1$ by modified SGSEL, which eliminate the two cliques expansion in SGSEL. According to DCMSP requirements, when the size constraint $p$ increases, DCSEL will search $F_2$ in $G \setminus F_1$. And then, DCSEL concatenates two subgraph and add some nodes which one of two endpoints is inside $F_1 \cup F_2$. \\
Specifically, DCSEL greedily removes nodes based on their total similarity $W(u)$ and records the best subgraph at the current size. While the size constraint $p$ decreasing, DCSEL check that is there exists any subgraph $F'$ such that $|F_1| > |F'| > p$ and $\alpha(F') > \alpha(F_1)$. Note that $|F_1|$ must larger than $|F|$, otherwise, $F'$ must return in first step. On the other hand, as the size constraint $p$ gradually increases to become larger than $F_1$, DCSEL will automatically search for $F_2$ in $G\setminus F$ and merge the two subgraphs together. And optional, greedyily select edges which one of the two end points is in $F_1 \cup F_2$. Note we call the last procedure is optional because we don't want the last procedure damage our objective value $\alpha(F_1\cup F_2)$.

\begin{algorithm}
\caption{DCSEL}
\begin{algorithmic}[1]
\Require Graph $G=(V, E, S)$, similarity constraint $s$, and size constraint $P=\{p_1, p_2, \cdots, p_k\}$
\Ensure 3-Approximation solution $F$
\State $A=\{-1, \cdots, -1\}, B=\{\varnothing, \cdots,\varnothing\} \triangleleft$ A, B denoted the objective value and optimal solution corresponding to their size.
\State $F_1, A, B \leftarrow$ Modified SGSEL$(G, s, p_1)$
\For {$i \leftarrow 2$ to $k$}
    \If {$p_{i} - p_{i-1} < 0$}
        \If {$\exists F \in B$, $\alpha(F) > \alpha(F_{i-1})$}
            \State $F_i\leftarrow F$
            
        \Else
            \State $F_i \leftarrow F_{i-1}$
        \EndIf
        \State continue;
    \EndIf
    \If {$p_i < |F_{i-1}|$}
        \State $F_{i} \leftarrow F_{i-1}$
        \State Continue;
    \EndIf
    \State $F_{i} \leftarrow$ Modified SGSEL$(G\setminus F_{i-1}, s, p_i-p_{i-1}) \cup F_{i-1}$
    \While {$I_{F_{i}}(u_i) > \alpha(F_{i}) $}
        \State $F_{i} \leftarrow F_{i} \cup \{u_i\}$
    \EndWhile
\EndFor
\State \Return $F_k$
\end{algorithmic}

\end{algorithm}

\section{Theoretical Analysis}
In this section, we will discuss the approximation ratio of our algorithm on EDCMSP. First, we define some notation in our proof. Next, we will discuss the problem in three cases. Last, we will show that our algorithm gives a $\frac{1}{3}$-approximation solution for both EDCMSP and DCMSP. 

In order to facilitate the subsequent proof and better elucidate our arguments, we need to define the following symbols:
\begin{enumerate*}[label=(\roman*)]
% \begin{enumerate}
    \item $F_i$ is the $\frac{1}{3}$-approximation solution of EDCMSP on the graph $G_i$ while considering the size constraint $p_i$ and the similarity constraint $s$.
    \item $OPT_i$ is the optimal solution of EDCMSP on the graph $G_i$.
    \item $\mathbb{F}$ is the optimal solution of EDCMSP on the graph $G_1\cup G_2$ while considering the size constraint $p_1 + p_2$ and the similarity constraint $s$.
    \item $\mathbb{F}_i$ refers to the set of nodes in $\mathbb{F}$ that belong to $G_i$. i.e., $\mathbb{F}_1 \subset G_1$ and $\mathbb{F}_1\cap G_2 = \varnothing$.
    \item $\mathbb{E}$ represents the collection of edges that have their two endpoints in different graphs.
% \end{enumerate}
\end{enumerate*}
In other words, $\mathbb{F}$ can be decomposed into three components, namely $\mathbb{F}_1$, $\mathbb{F}_2$, and $\mathbb{E}$. This observation inspires us to establish individual bounds for each of these components.

\begin{algorithm}
\scriptsize
\caption{Modified SGSEL}
\begin{algorithmic}[1]
\Require Graph $G=(V, E, S)$, similarity constraint $s$, and size constraint $p$
\Ensure 3-Approximation solution $F$
\State $F^{APX} \leftarrow \varnothing \triangleleft$ Initialize the approximation solution F.
\State $F \leftarrow V$
\State $A=\{-1, \cdots, -1\}, B=\{\varnothing, \cdots,\varnothing\}$ 
\While{ $F \neq \varnothing$}
    \State $v\leftarrow \argmin_{u\in F} I_{F}(u)$ $\triangleleft$ $I_{F}(u)$ is defined as the sum of the edge weights between vertex $u$ and all other vertices $t$ in $F$.
    \State $F\leftarrow F\setminus\{v\}$
    \If{$\alpha(F) > \alpha(F^{APX}) $}
        \State $F^{APX} \leftarrow F$
    \EndIf
    \If {$\alpha(F) > A[|F|]$}
        \State $A[|F|]\leftarrow \alpha(F)$
        \State $B[|F|]\leftarrow F$
    \EndIf
\EndWhile
\State \Return $F^{APX}, A, B$
\end{algorithmic}
\end{algorithm}

\begin{lemma} \label{easiest_case}
    Suppose $\alpha(\mathbb{F}) \leq \alpha(OPT_1)$ and $\alpha(\mathbb{F}) \leq \alpha(OPT_2)$. Then $F_1 \cup F_2$ is a $\frac{1}{3}$-approximation solution of EDCMSP. i.e., $\alpha(F_1 \cup F_2) \geq \frac{1}{3}\alpha(\mathbb{F}) $. 
\end{lemma}
\begin{proof}
This can easily prove by analysis the relationship between $F_1 \cup F_2$ and $\mathbb{F}$. 
\begin{align*}
\alpha(F_1 \cup F_2) &= \frac{W(F_1) + W(F_2) + \Delta}{|F_1| + |F_2|} 
                     > \frac{W(F_1) + W(F_2) }{|F_1| + |F_2|}  \\  
                     &> \frac{\frac{1}{3}\alpha(OPT_1)|F_1|+\frac{1}{3}\alpha(OPT_2)|F_2|}{|F_1| + |F_2|} 
                     \geq \frac{\frac{1}{3} \alpha(\mathbb{F}) |F_1|+\frac{1}{3} \alpha(\mathbb{F}) |F_2|}{|F_1| + |F_2|} 
                     = \frac{1}{3}\alpha(\mathbb{F})
\end{align*}
Note that $\Delta > 0$ is the total edge weight whose two endpoint in two different graphs $G_1$ and $G_2$. The third inequality is because $W(F_1) = \alpha(F_1)|F_1| \geq \frac{1}{3}\alpha(OPT_1)|F_1|$, so does $W(F_2) \geq \frac{1}{3}\alpha(OPT_2)|F_2|$.

\end{proof}

\begin{lemma} \label{weighted_sum}
    Suppose we have two disjoint subgraph of graph $G$, say $A$ and $B$. If $\alpha(A) < \alpha(B)$, then $ \alpha(A)<\alpha(A\cup B)$.
\end{lemma}
\begin{proof}
\begin{align*}
\alpha(A \cup B) &= \frac{W(A) + W(B) + \Delta_{AB}}{|A| + |B|}
                 &> \frac{\alpha(A)|A| + \alpha(B)|B|}{|A| + |B|}
                 &> \frac{\alpha(A)|A| + \alpha(A)|B|}{|A| + |B|}\\
                 &= \alpha(A)
\end{align*}
Note that $\Delta_{AB} > 0$ is the total edge weight whose two endpoint in two different graphs $A$ and $B$.

\end{proof}

\begin{theorem} 
    Given a graph $G$ with similarity constraint $s$ and size constraint $p$, and a $\frac{1}{3}-$approximation solution $F_1$ given by SGSEL. Given any Dynamic size constraint $\Delta p$, the DCSEL will return a solution with ratio $\frac{1}{3}$.
\end{theorem}
\begin{proof} First of all, if the dynamic constraint $\Delta p$ is small than $0$, it indicates a loosening of the overall constraints. Since the trimming process of original constraint $\{s, p\}$ also holds the ratio $\frac{1}{3}$. Thus,there is no necessity for additional deliberation on this scenario.\\\\
For $\Delta p > 0$, We can categorize all scenarios into three disjoint cases.
\begin{enumerate}
    \item $\mathbb{F} \cap G_1 = \varnothing$, or say $\mathbb{F} \subseteq G_2$.
    \item $\mathbb{F} \cap G_2 = \varnothing$, or say $\mathbb{F} \subseteq G_1$.
    \item $\mathbb{F} \cap G_1 \neq \varnothing$ and $\mathbb{F} \cap G_2 \neq \varnothing$.
\end{enumerate}
WLOG, if $\mathbb{F} \cap G_1 = \varnothing$, then we have $\alpha(\mathbb{F}) \leq \alpha(OPT_2)$. This is because $\mathbb{F} \subseteq G_2$ and $|\mathbb{F}| > p_2$. Follow the definition of $OPT_2$, $\alpha(\mathbb{F})$ must smaller than $\alpha(OPT_2)$. If $\alpha(\mathbb{F}) \leq \alpha(OPT_1)$, we have proved in $\cref{easiest_case}$. Thus, we can suppose $\alpha(\mathbb{F}) > \alpha(OPT_1)$. On the other hand, for all node $u$ belong to graph $G_1$, we can deduce that $I_\mathbb{F}(u) < \alpha(\mathbb{F})$. If not, $\alpha(\mathbb{F} \cup \{u\})=\frac{W(\mathbb{F}) + I_\mathbb{F}(u) }{|\mathbb{F}| + 1}=\frac{|F|}{|F|+1}\alpha(\mathbb{F}) + \frac{1}{|F|+1} I_\mathbb{F}(u) > \frac{|F|}{|F|+1}\alpha(\mathbb{F}) + \frac{1}{|F|+1} \alpha(\mathbb{F}) = \alpha(\mathbb{F})$ which imply $\mathbb{F} \cup \{u\}$ is better than $\mathbb{F}$ $\rightarrow\!\leftarrow$. Thus, this suggest us can greedily choose the node $t\in G_2$ by $I_{F_2}(t) > \alpha(F_2)$ until satisfy the size constraint $p_1 + p_2$. Since $\mathbb{F} \subseteq G_2$, there must exists enough node that can form the approximation solution.
\end{proof}
Now, we consider the third case: $\mathbb{F} \cap G_1 \neq \varnothing$ and $\mathbb{F} \cap G_2 \neq \varnothing$. Note that $\mathbb{F} = \mathbb{F}_1 \cup \mathbb{F}_2 \cup \mathbb{E}$. As mention before, we will estimate the upper bound each by each. First, without loss of generality, we show $\alpha(\mathbb{F}_1) < 3 \alpha(F_1)$. If $\alpha(\mathbb{F}_1)\leq\alpha(OPT_1)$. By the property of $F_1$, we have $\alpha(\mathbb{F}_1)\leq\alpha(OPT_1)<3\alpha(F_1)$. on the other hand, we suppose $\alpha(\mathbb{F}_1)>\alpha(OPT_1)$. If $\mathbb{F}_1 \cap \alpha(OPT_1) = \varnothing$, then $\mathbb{F}_1 \cup \alpha(OPT_1)$ form a solution better than optimal solution $\alpha(\mathbb{F}_1)$ by $\cref{weighted_sum}$, which lead to contradiction. If $\mathbb{F}_1 \cap \alpha(OPT_1) \neq \varnothing$, say $\mathbb{F}_1 \cap OPT_1 = R$ and $\mathbb{F}_1\setminus OPT_1 = R'$. For any subset $U$ of $R'$, we can claim that $\frac{I_R(U)}{|U|} = \alpha(U)<\alpha(OPT_1)$. Otherwise, $OPT_1 \cup U$ will form an optimal solution on $G_1$, which lead to contradiction. Next, we discover that there exists a subset $V$ belong to $OPT_1 \setminus \mathbb{F}_1$ such that $I_R\setminus V(V) > I_R\setminus V(R')$, where $|V|=|R'|$. If not, replacing $R'$ with $V$ can yield another optimal solution, which lead to contradiction. Finally, we can use this $V$ to form a better component in EDCMSP, say $\mathbb{F}_1 \cup V$, which leads a contradiction. Similarily, we can say $\alpha(\mathbb{F}_2) < 3\alpha(F_2)$. Last of all, we are going to show that our algorithm picking a node set, denoted by $\Delta$, is an upper bound of $W(\mathbb{E})$. Since $\Delta$ is composed by greedily picking up the node in those which one of the node has already in $F_1$, but not $F_2$, and vice versa. Thus, we set a stopping criteria that can bound $\frac{1}{3}W(\mathbb{E})$ and not decrease the average similarity. Therefore, each component of $\mathbb{F}$ is bound by our algorithm selected subgraph. In other words, our algorithm is a $\frac{1}{3}$-approximation algorithm.   
\begin{lemma} \label{positive_delta_s}
    Given a graph $G$ with similarity constraint $s$ and size constraint $p$, and a $\frac{1}{3}-$approximation solution $F_1$ given by SGSEL. If Dynamic similarity constraint $\Delta s>0$, then DCSEL will return a solution with ratio $\frac{1}{3}$.
\end{lemma}
\begin{proof}
Let $U$ be the set of nodes whose edges similarity are all lower than $s + \Delta s$ in $G$. i.e., our search space of algorithm reduced to $G\setminus U$. Since the increment of similarity constraint eliminate nodes whose has lower edge similarity, thus, it doesn't impact the optimal solution. However, there might be some nodes are in both $U$ and $F_1$. If $|F_1\setminus U| \geq p$, we don't need to modify our result. $F_1\setminus U$ still hold the ratio $\frac{1}{3}$ and the two constraints of our member selection problem. On the other hand, if $|F_1\setminus U|<p$, we need run our DCSEL again to generate approximation solution $F$. In this step, the ratio is still hold, since the procedure of new $F$ is same as the initial constraints \{$|F_1\setminus U|, s+\Delta s$\}, size constraint and similarity constraint respectively, and then with dynamic constraint $\Delta p = p - |F_1\setminus U|$.

\end{proof}

\begin{lemma} \label{negative_delta_s}
    Given a graph $G$ with similarity constraint $s$ and size constraint $p$, and a $\frac{1}{3}-$approximation solution $F_1$ given by SGSEL. If Dynamic similarity constraint $\Delta s<0$, then DCSEL will return a solution with ratio $\frac{1}{3}$.
\end{lemma}
\begin{proof}
    Let $U$ be the set of nodes whose edges similarity are all higher than $s + \Delta s$, but lower than $s$ in graph $G$. i.e., $\forall u \in U,$ $s+\Delta s < e(u) < s$. Since $U$ can be seen as a new search space of DCSEL, we can easily run the DCSEL once, adding some nodes to $F_1$ from $(G\setminus F_1)\cup U $ to form new approximation solution $F$. The ratio is hold because the procedure of new $F$ is same as the initial constraints $\{p,$ $ s+\Delta s\}$.
\end{proof}

\begin{theorem} 
    Given a graph $G$ with similarity constraint $s$ and size constraint $p$, and a $\frac{1}{3}-$approximation solution $F_1$ given by SGSEL. Given any Dynamic similarity constraint $\Delta s$, the DCSEL will return a solution with ratio $\frac{1}{3}$.
\end{theorem}
\begin{proof} 
    The \cref{positive_delta_s} and \cref{negative_delta_s} demonstrate the all cases of similarity constraint adjustment.
\end{proof}

\section{Extra Nodes of Graph Adjustment}
In order to deal with the constantly changing node instances in the real-world network, we have also discussed the scenarios involving changes to nodes on the graph. The problem can be formulated by the following statement: 
\begin{theorem} 
    Giving a graph $G$ with similarity constraint $s$ and size constraint $p$, and a $\frac{1}{3}-$approximation solution $F_1$ given by SGSEL. If there exists another graph $G'$, then our DCSEL algorithm can find an approximation solution with approximation ration $\frac{1}{3}$ on $G\cup G'$.
\end{theorem}
\begin{proof}
     Since $G'$ can be seen as a new search space of DCSEL, we can easily run the DCSEL, which adding some nodes to $F_1$ from $(G\cup G') \setminus F_1$ to form new approximation solution $F$. The ratio is hold because the procedure of $F$ is same as the initial constraints $\{p,$ $ s\}$ on $G \cup G'$.
\end{proof}
\begin{figure}[ht]
\centering
    \begin{subfigure}[b]{.6\linewidth}
        \centering
        \includegraphics[width=\linewidth]{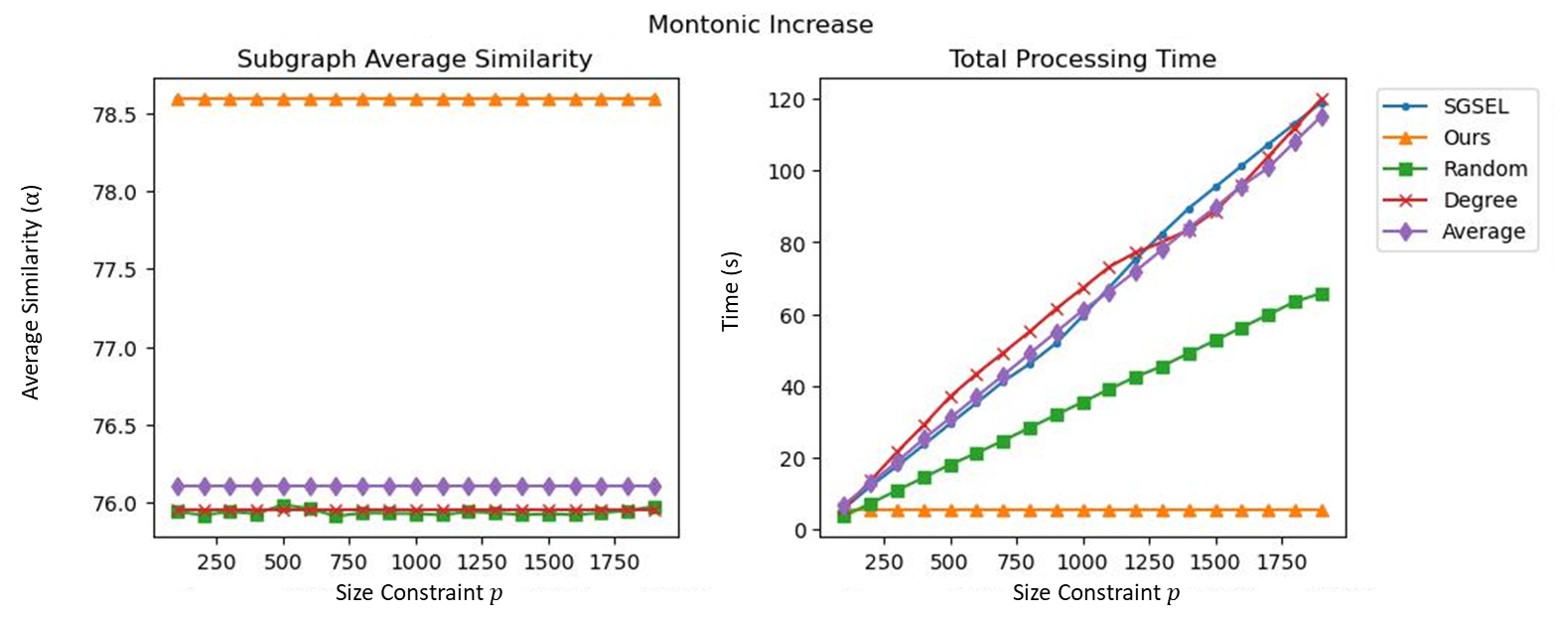}
        \caption{Dataset: Cora}
        \label{montonic_increase_cora}
     \end{subfigure}\\
    \begin{subfigure}[b]{.6\linewidth}
        \centering
        \includegraphics[width=\linewidth]{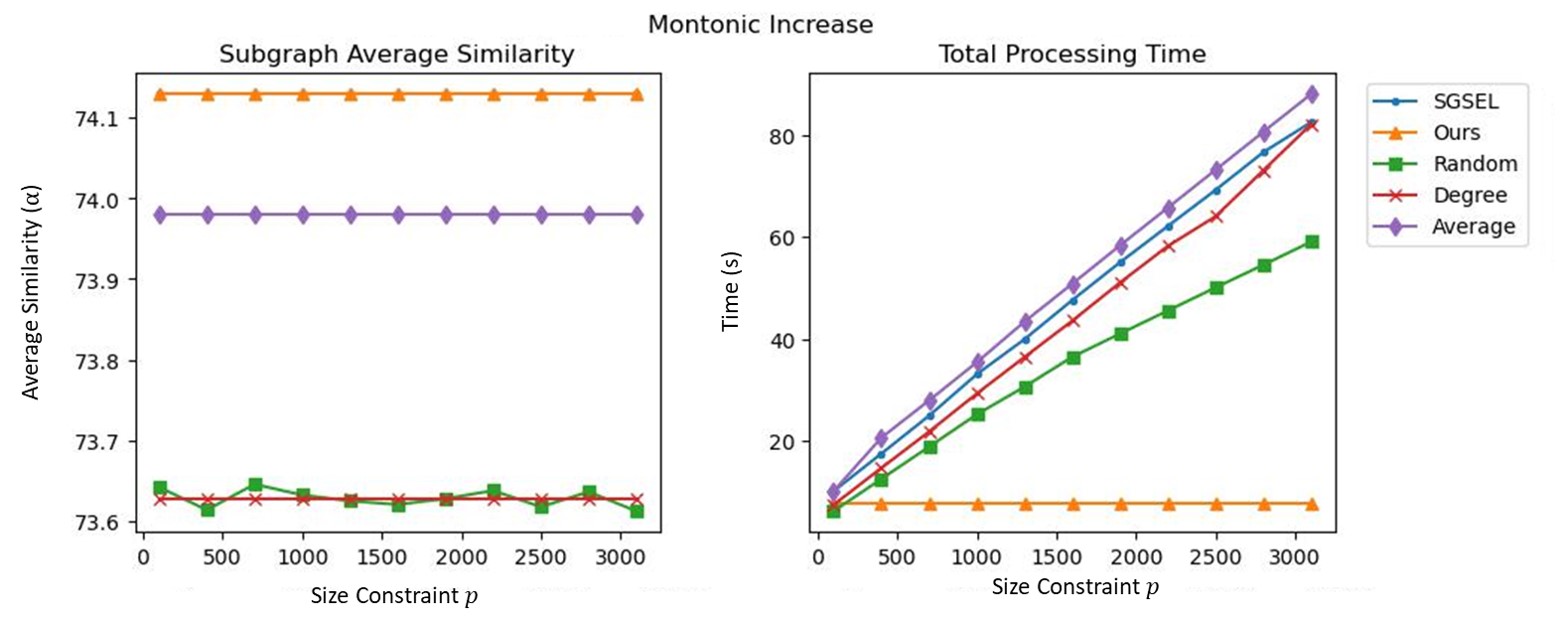}
        \caption{Dataset: Citeseer}
        \label{montonic_increase_citeseer}
    \end{subfigure}\\
    \begin{subfigure}[b]{.6\linewidth}
        \centering
        \includegraphics[width=\linewidth]{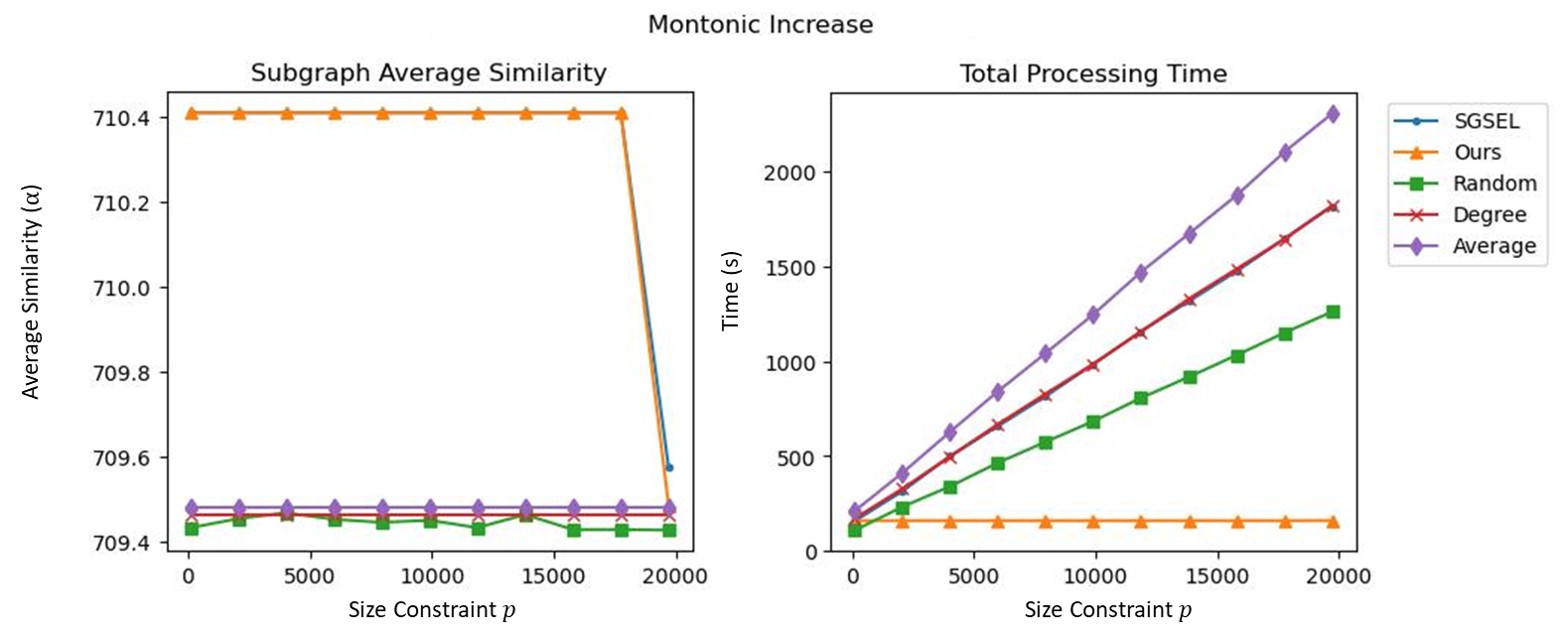}
        \caption{Dataset: Pubmed}
        \label{montonic_increase_pubmed}
    \end{subfigure}
    \caption{Size constraint $p$ monotonic increase}
    \label{montonic_increase}
\end{figure}

\begin{figure}[ht]
\centering
    \begin{subfigure}[b]{.6\linewidth}
        \centering
        \includegraphics[width=\linewidth]{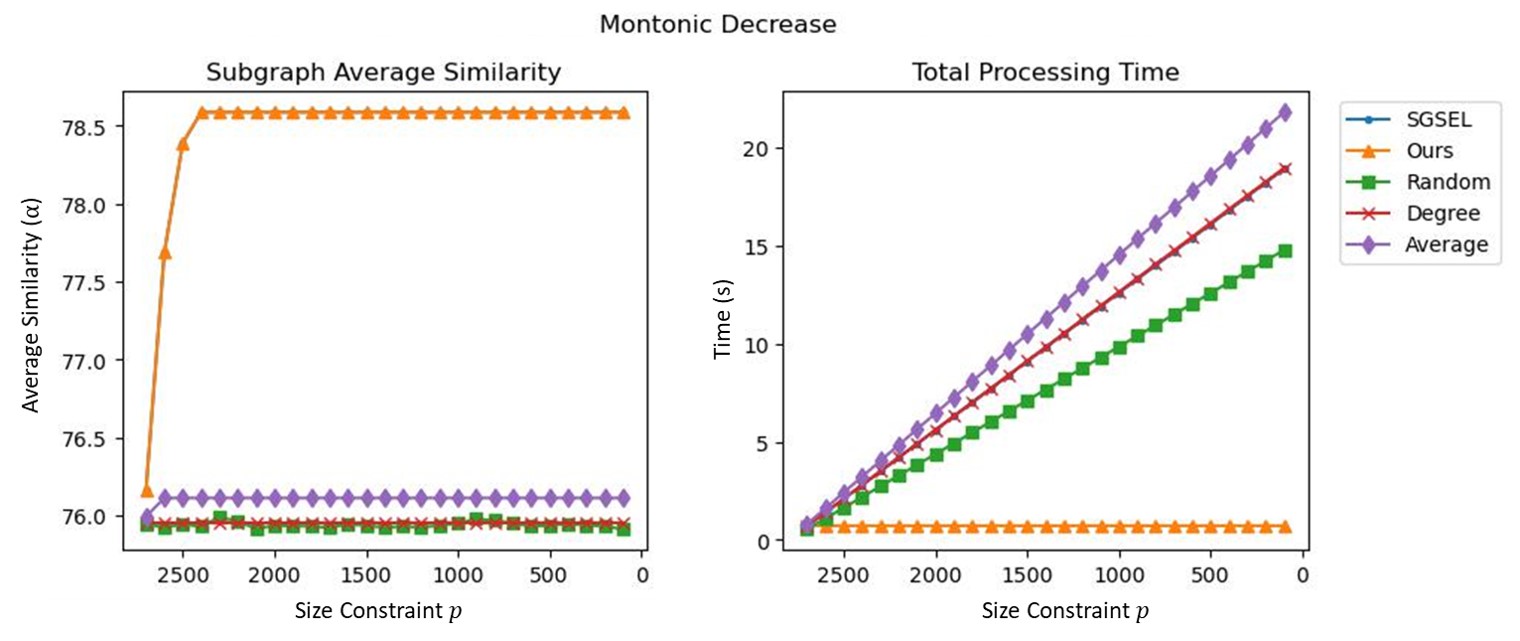}
        \caption{Dataset: Cora}
        \label{montonic_decrease_cora}
     \end{subfigure}\\
    \begin{subfigure}[b]{.6\linewidth}
        \centering
        \includegraphics[width=\linewidth]{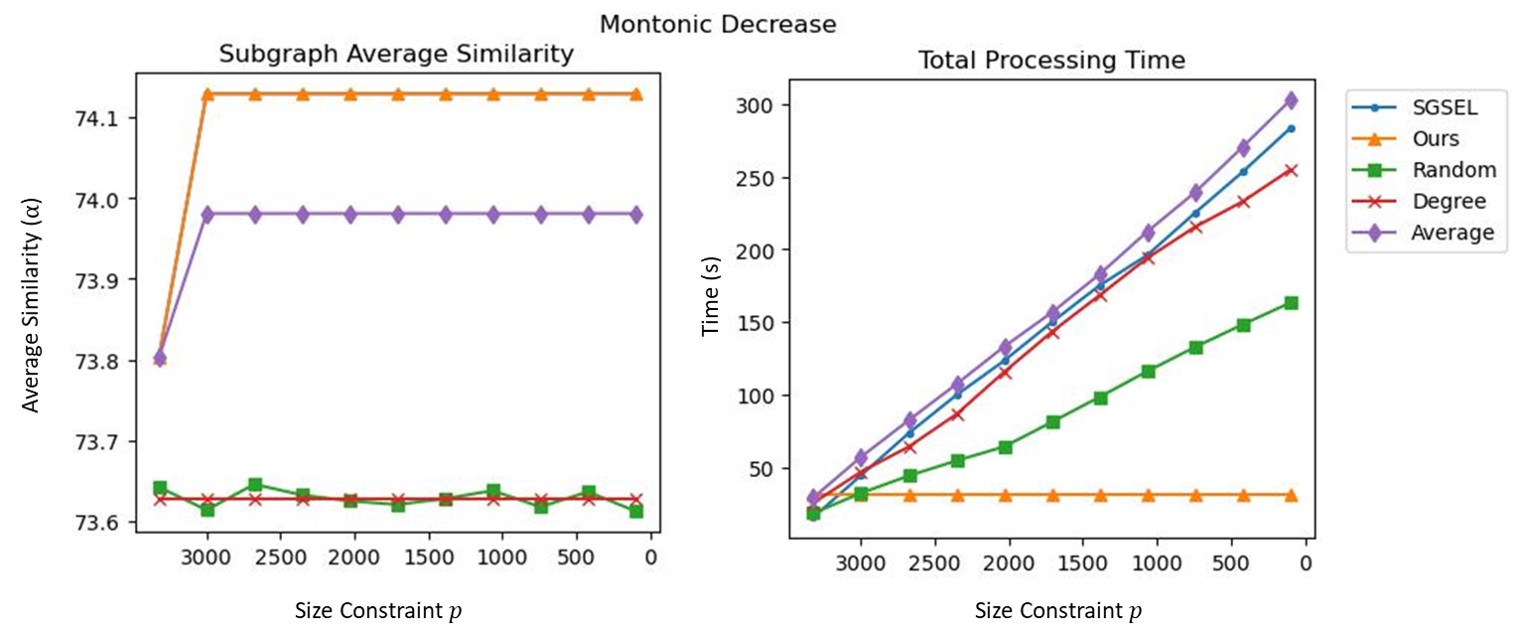}
        \caption{Dataset: Citeseer}
        \label{montonic_decrease_citeseer}
    \end{subfigure}\\
    \begin{subfigure}[b]{.6\linewidth}
        \centering
        \includegraphics[width=\linewidth]{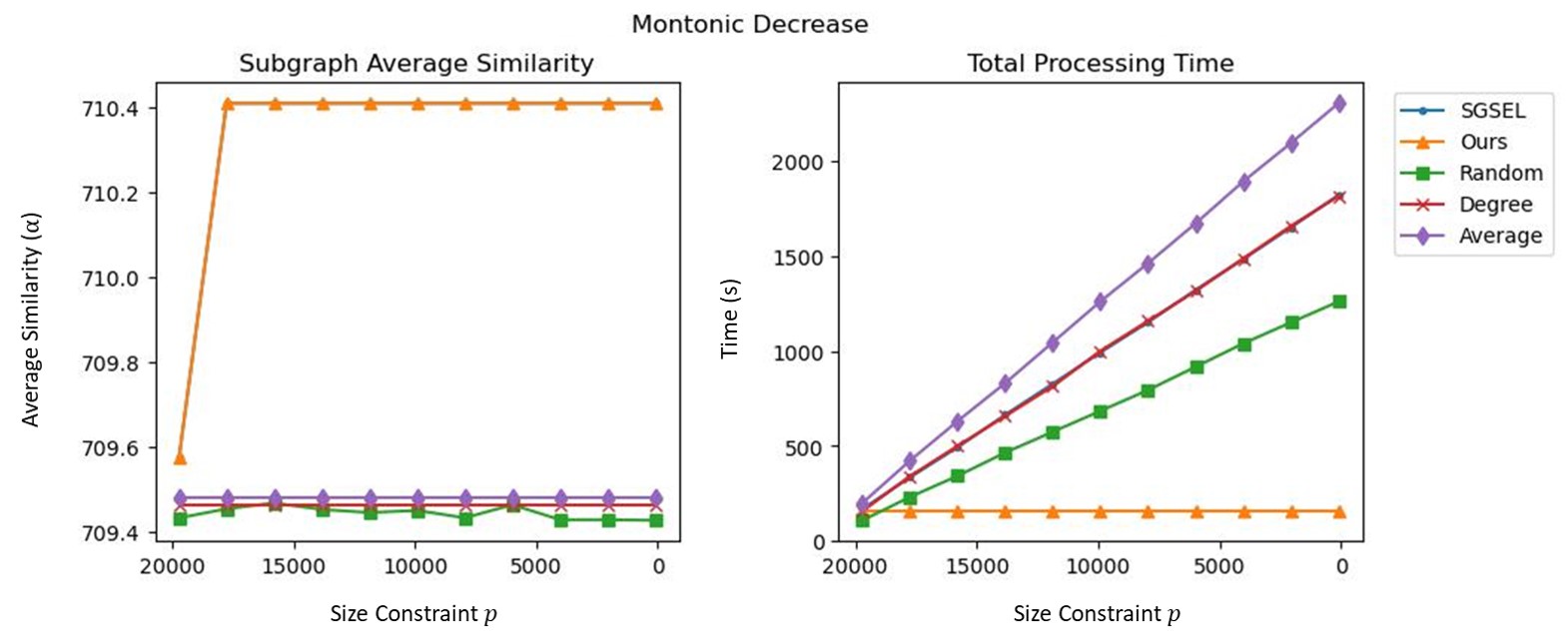}
        \caption{Dataset: Pubmed}
        \label{montonic_decrease_pubmed}
    \end{subfigure}
    \caption{Size constraint $p$ montonic decrease}
    \label{montonic_decrease}
\end{figure}

\begin{figure}[ht]
\centering
    \begin{subfigure}[b]{.6\linewidth}
        \centering
        \includegraphics[width=\linewidth]{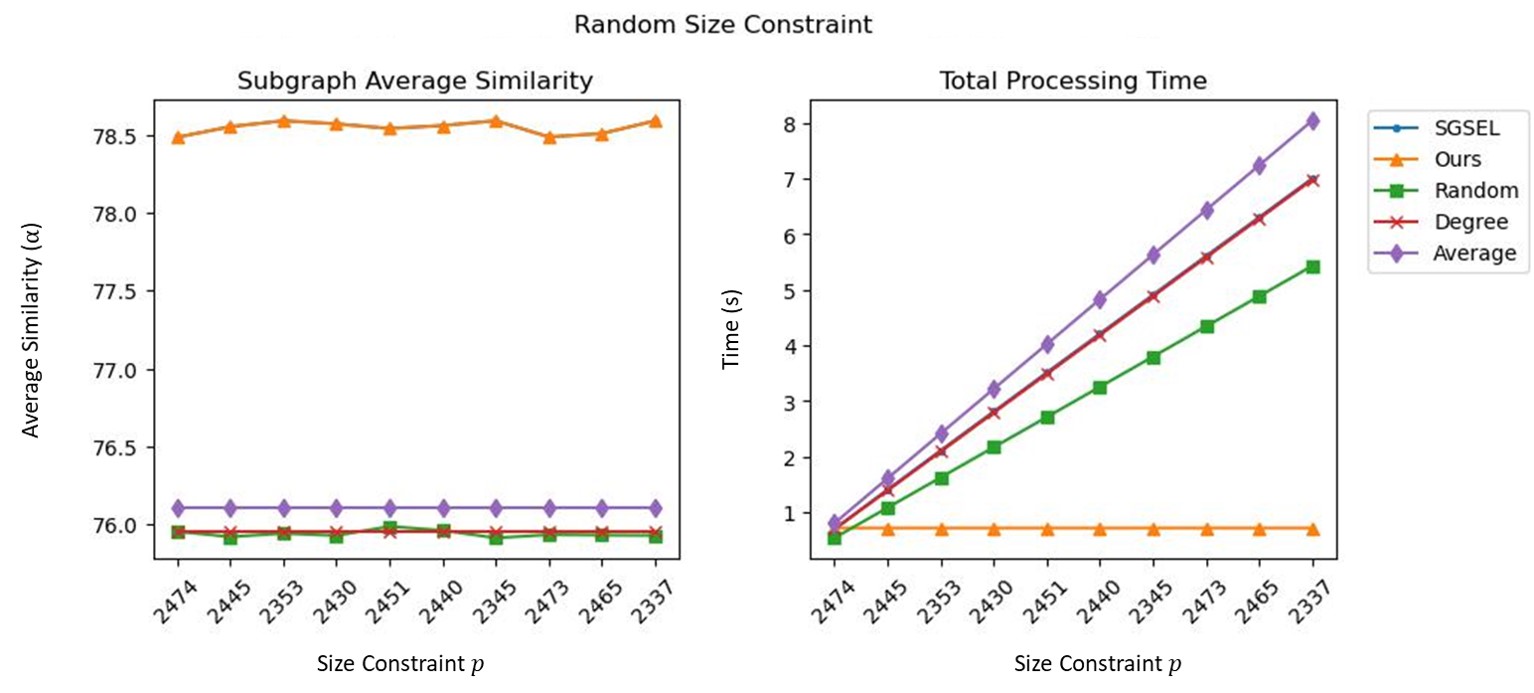}
        \caption{Dataset: Cora}
        \label{random_p_cora}
     \end{subfigure}\\
    \begin{subfigure}[b]{.6\linewidth}
        \centering
        \includegraphics[width=\linewidth]{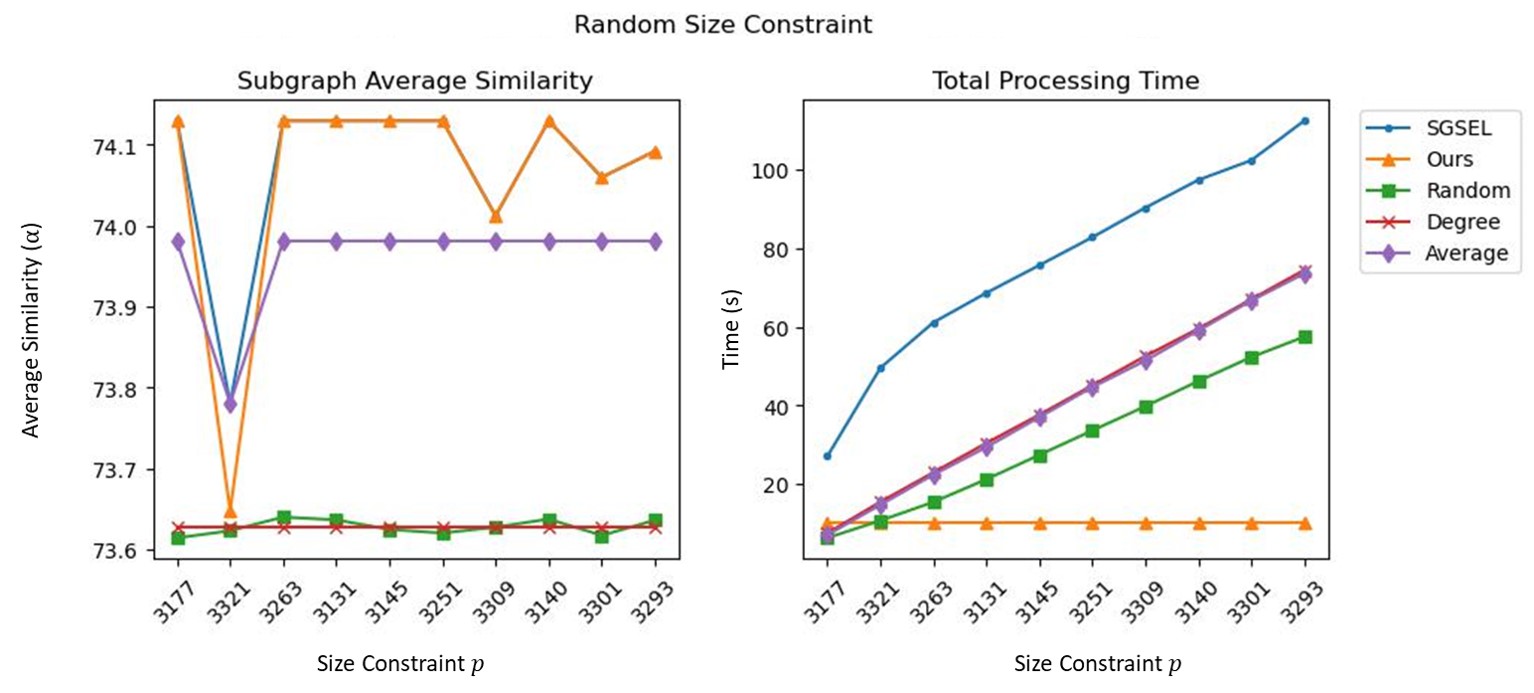}
        \caption{Dataset: Citeseer}
        \label{random_p_citeseer}
    \end{subfigure}\\
    \begin{subfigure}[b]{.6\linewidth}
        \centering
        \includegraphics[width=\linewidth]{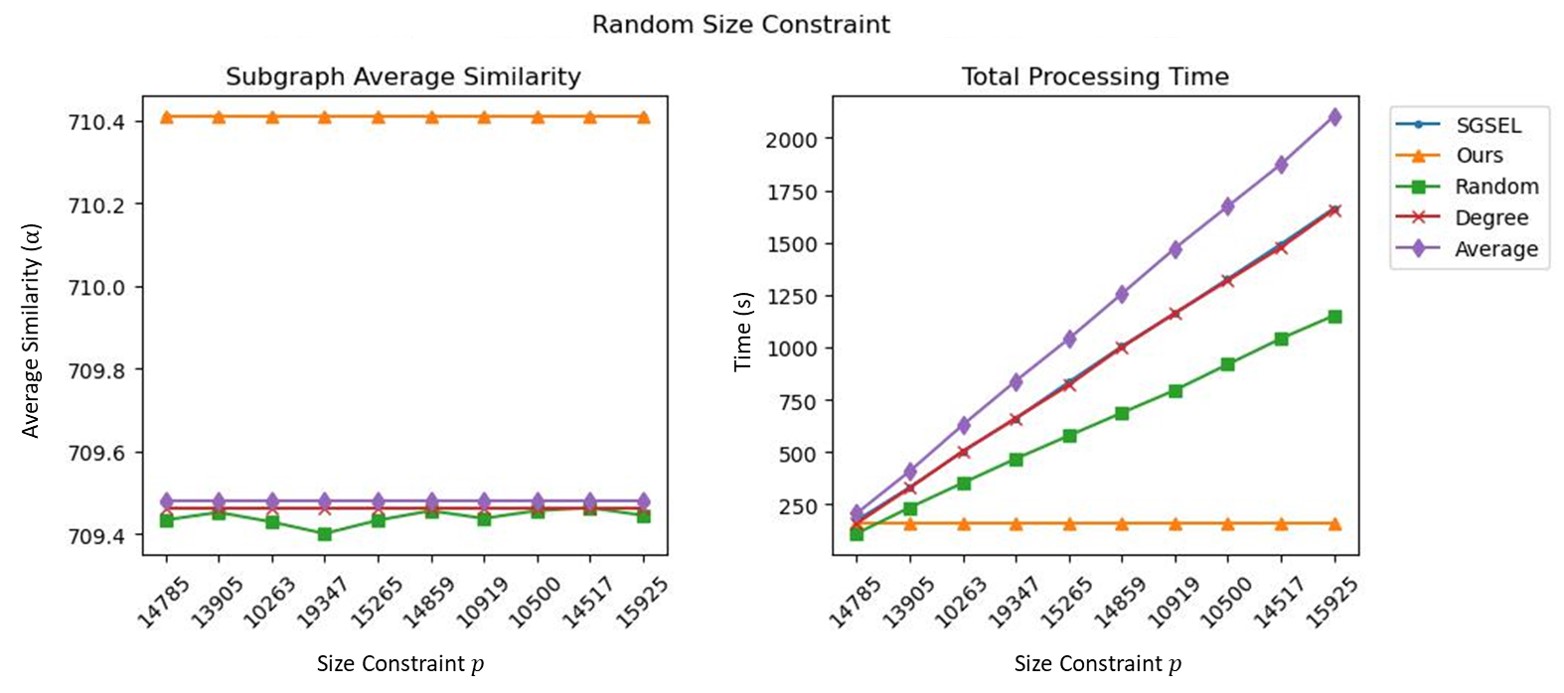}
        \caption{Dataset: Pubmed}
        \label{random_p_pubmed}
    \end{subfigure}
    \caption{Size constraint $p$ randomly selected}
    \label{randon_select}
\end{figure}

\section{Experiments}
In this section, we testify the effectiveness of our DCSEL algorithm, we compare it with four baselines. Our experiment is conducted on R9-5900x, RTX 3090 Ti, $4\times16$ GB memory, and Ubuntu 22.04.
\subsection{Datasets}
We employ three widely-adopted benchmark datasets for our experiment. Cora, Citeseer, and Pubmed are three citation network dataset. Node represent documents and edges represent citation links. Node features of Cora and Citeseer are documents encoded by bag-of words feature vector, while the features of Pubmed are TF/IDF weighted word frequency. Their statistic is summarized in Table \ref{basic_statistics}
\begin{table}[h]
\centering
\caption{Dataset Statistics}
\scriptsize
\begin{tabular}{||c c c c||} 
 \hline
 Dataset & \#Nodes & \#Edges & \#Features \\ [0.5ex] 
 \hline\hline
 Cora & 2,485 & 5,069 & 1,433 \\ 
 Citeseer & 2,110 & 3,668 & 53,703 \\
 Pubmed & 19,717 & 44,324 & 500 \\
 [1ex] 
 \hline
\end{tabular}
\label{basic_statistics}
\end{table}

\subsection{Baseline}
To demonstrate the effectiveness of our proposed approach, we compare our DCSEL with other four methods, including SGSEL, Random, Degree, and Average.
\begin{enumerate*}[label=(\Roman*)]
    \item SGSEL, which was introduced in Section \ref{Preliminaries} in detail.
    \item Random. Randomly remove nodes in each iteration and return maximum average similarity subgraph in whole epoch.
    \item Degree. Eliminate nodes based on their average incident similarity, i in each iteration, and then calculate the maximum average similarity throughout the entire epoch.
    \item Average. In each iteration, remove nodes based on their average incident similarity, i.e.,$\frac{I_{\Gamma_i}(u)}{|\mathcal{N}(u)|}$, and then calculate the maximum average similarity over the entire epoch.
\end{enumerate*}
\subsection{Result}
We evaluate our method under three different size constraint $p$ changes, including
\begin{enumerate*}[label=(\roman*)]
    \item monotonic increase, in figure \ref{montonic_increase}
    \item monotonic decrease, in figure \ref{montonic_decrease}
    \item random select, in figure \ref{randon_select}
\end{enumerate*}. We first compare DCSEL with other four baselines. We can find out that DCSEL achieve objective value as high as SGSEL in most scenario. This is because DCSEL has same procedure as SGSEL in beginning, thus it can extract the same induced graph. Although our proposed DCSEL is concatenate two induced subgraph, experiment result show DCSEL can find the approximation solution properly. In addition, our algorithm processes dynamic constraint member selection problem efficiently. This might because $F_1$ is large enough, i.e., $|G\setminus F_1| \ll |F_1|$. This observation explain that how does DCSEL work efficiently.

\bibliographystyle{splncs04}

\bibliography{refer.bib}

\end{document}